\documentclass[letterpaper, 10 pt, conference]{ieeeconf}
\IEEEoverridecommandlockouts
\usepackage{cite}
\usepackage{amsmath,amssymb,amsfonts}
\usepackage{algorithm}
\usepackage{algorithmic}
\usepackage{graphicx}
\usepackage{textcomp}
\usepackage[dvipsnames,svgnames]{xcolor}
\def\BibTeX{{\rm B\kern-.05em{\sc i\kern-.025em b}\kern-.08em
        T\kern-.1667em\lower.7ex\hbox{E}\kern-.125emX}}
\usepackage{comment}
\usepackage{enumerate}
\usepackage{paralist}

\newcommand{\cov}[3]{\boldsymbol{k}_{#1,#2}^{(#3)}}

\newcommand{\x}{\boldsymbol{x}}

\newcommand{\qedsymbol}{$\blacksquare$}
\newcommand{\bm}{\boldsymbol}
\newcommand{\safeopt}{{\textsc{SafeOpt}}}
\newcommand{\safeslope}{{\textsc{SafeSlope}}}
\newcommand{\safeucb}{{\textsc{SafeUCB}}}
\newtheorem{theorem}{Theorem}[section]

\newtheorem{remark}{Remark}[section]

\newtheorem{corollary}{Corollary}[section]
\newtheorem{lemma}{Lemma}[section]

\begin{document}
    \title{\LARGE A Multi-Fidelity Bayesian Approach to Safe Controller Design\\
        \small Extended Version with Pseudocode, Extensions, and Full Proofs \thanks{This work was supported in part by ARO grant W911NF-18-1-0325 and in part by NSF Award CNS-2134076.}}

    \author{Ethan Lau \and Vaibhav Srivastava \and Shaunak D.~Bopardikar \thanks{The authors are with the Electrical and Computer Engineering Department at Michigan State University.}}

    \maketitle
    \bstctlcite{IEEEexample:BSTcontrol}

    \begin{abstract}
        Safely controlling unknown dynamical systems is one of the biggest challenges in the field of control systems. Oftentimes, an approximate model of a system's dynamics exists which provides beneficial information for control design. However, differences between the approximate and true systems present challenges as well as safety concerns. We propose an algorithm called \safeslope\ to safely evaluate points from a Gaussian process model of a function when its Lipschitz constant is unknown. We establish theoretical guarantees for the performance of \safeslope\ and quantify how multi-fidelity modeling improves the algorithm's performance. Finally, we present a case where \safeslope\ achieves lower cumulative regret than a naive sampling method by applying it to find the control gains of a linear time-invariant system.
    \end{abstract}

    \section{Introduction}
    In the realm of control systems, there exist many instances in which the dynamics are not fully modeled. While an approximation of the dynamics may exist, variations in the system's components or environment may cause the system to deviate from the design model. For example, consider off-the-shelf robotics kits. Though identically designed, each robot possesses variations that cause its performance to vary from the design model. In this case, we can consider each robot to be a \emph{black-box system}, possessing accessible input-output data but inaccessible exact dynamics. We study how the true system output can be used with a design or simulated model to create an improved model of the true dynamical system.

    Gaussian process (GP) regression is a popular non-parametric technique for optimizing unknown or difficult-to-evaluate cost functions. The upper confidence bound (UCB) algorithm \cite{srinivas2012information} guarantees asymptotic zero regret when iteratively sampling a GP. Multi-fidelity Gaussian processes (MF-GPs) predict a distribution from multiple correlated inputs.
    The linear auto-regressive (AR-1) model is an MF-GP that uses a cheaper model to assist in evaluating a more complex model \cite{kennedy2000predicting}. The AR-1 model's recursive structure allows it to effectively model correlated processes while its decoupled form enables computationally efficient parameter learning. Analytical guarantees have also been established when applying Bayesian optimization to MF-GPs \cite{kandasamy2019multi, song2019general}.

    Recently, GPs have been explored for control design.
    GPs and MF-GPs have been applied to finding ideal control gains for linear time-invariant (LTI) systems \cite{marco2017design, marco2017virtual}.
    MF-GPs have also been applied to falsification frameworks for testing system safety \cite{shahrooei2022falsification}.
    However, these papers primarily contain experimental results, without any mathematical guarantees for the approach.

    Other data-driven methods have been proposed to control LTI systems.
    Model-based approaches reconstruct a model of the system dynamics from trajectories of similar systems \cite{oymak2019non, xin2023learning} and have been studied for robustness \cite{zheng2021sample}.
    When data is abundant, model predictive control may be used to find an ideal control strategy \cite{hewing2020learning}. Model-free approaches aim to directly control a system without learning the system dynamics \cite{baggio2019data, de2019formulas, sun2021learning}.

    Whether model-based or model-free, a critical aspect of controller design is safety. A recent review of safe learning in control classifies approaches based on the strength of the safety guarantee and the required knowledge of the system's dynamics \cite{brunke2022safe}. An ideal approach ensures strict constraints are met for a system with unknown dynamics. Despite proposed solutions, there is a gap in work involving using GPs for safe control design.

    We consider a data-driven Bayesian optimization approach to find optimal controllers of black-box systems. The following are our \underline{main contributions}:

    1) We establish \safeslope, a safe exploration algorithm with analytical bounds when the Lipschitz constant of a black-box cost function is unknown. Unlike \safeopt\ \cite{sui2015safe}, which relies on a known Lipschitz constant, we upper bound the slope using the posterior distribution of the GP.

    2) We formalize how an AR-1 model can improve the choice of inputs. In particular, we show how its conditional covariance matrix can be used to reduce the upper bound on the information gain. We also numerically compare the performance of an AR-1 model to a single-fidelity GP.

    \section{Problem Overview}
    \subsection{Motivating Scenario}\label{sec:scenario}
    For this problem, we model a true system with LTI dynamics, $\bm{z}_{j+1} = A\bm{z}_{j} + B\bm{u}_{j}$,
    %\begin{align}\label{eq:state_space}
    %    \bm{z}_{j+1} = A\bm{z}_{j} + B\bm{u}_{j},
    %\end{align}
    where $\bm{z}\in\mathbb{R}^n$ is the state, $\bm{u}\in\mathbb{R}^{p}$ is the input, and $A\in\mathbb{R}^{n \times n}$, $B\in\mathbb{R}^{n \times p}$ are the system matrices. Under  feedback control, the system input is $\bm{u}_j = -K\bm{z}_j$, where $K \in \mathbb{R}^{p \times n}$ is the control gain. Given an initial state $\bm{z}_0$ and weighting matrices $Q$ and $R$, the system's infinite-horizon LQR cost for a set of gains $K$ is
    %----------
    \begin{align} \label{eq:lqr_cost}
        &J(K) {=} \sum_{j=0}^{\infty } \bm{z}_0^T(A{-}BK)^{Tj} [Q{+}K^TRK](A{-}BK)^{j}\bm{z}_0.
    \end{align}
    %----------
    Our goal is to minimize \eqref{eq:lqr_cost} by finding the ideal gain $K^*$.\footnote{We demonstrate the algorithm on an LTI system with quadratic cost for simplicity's sake. However, our algorithm may also be applied to any system possessing a parameterized controller with a measurable performance metric.}
    \smallskip

    When $A$ and $B$ are \emph{unknown}, determining an ideal $K^*$ becomes more challenging.
    We consider a situation in which a design model of the system has the evolution $\bm{z}_{j+1} = \hat{A}\bm{z}_{j} + \hat{B}\bm{u}_{j}$ and associated cost $\hat{J}$,
    with $\hat{A}\in\mathbb{R}^{n \times n}$, $\hat{B}\in\mathbb{R}^{n \times p}$. The design model has the same dimension as the true system, but its entries differ from those in the true system. We aim to leverage the design model
    to quickly find an ideal $K^*$ while avoiding gains that cause instability.

    We propose using an MF-GP framework that \emph{only requires the input-output data} from the auxiliary and the true systems. Here, the input is the choice of gain $K$, and the output is $J(K)$. We apply an AR-1 model by treating $(\hat{A},\hat{B})$ and $(A,B)$ as the low- and high-fidelity models, respectively. By using a search algorithm that guarantees safety, we seek to avoid sampling unstable controller gains.

    \subsection{Multi-Fidelity Gaussian Processes (MF-GPs)}
    A Gaussian process is a collection of random variables such that every finite set of random variables has a multi-variate Gaussian distribution \cite{rasmussen2006gaussian}.
    A GP is defined over a space $\mathcal{X} \subset \mathbb{R}^{n}$ by its mean function $\mu\,:\,\mathcal{X} \to \mathbb{R}$ and its covariance (kernel) function $k\,:\,\mathcal{X}\times \mathcal{X} \to \mathbb{R}$.

    Given a set of points $\bm{X}_t = \{\x_1,\dots, \x_t\}$, we create a covariance matrix $\bm{k}(\bm{X}_t, \bm{X}_{t}) = [k(\x_i,\x_j)]_{i,j=1}^{t,t}$, which is always positive definite. The covariance between a point and a set of points yields a covariance vector $\bm{k}(\x) := \bm{k}(\bm{X}_t,\x) = [k(\x_1, \x) \ldots k(\x_t,\x)]^T$.

    Let $f$ be a sample from a GP with mean $\mu$ and kernel $k$.
    Suppose we have prior data $\bm{X}_t$ and $\bm{Y}_t = \{y_1,\dots, y_t\}$, where $y_i = f(\x_i) +\eta$ has measurement noise $\eta\sim N(0,\xi^2)$. Then the posterior distribution of $f$ at $\x$ is a normally distributed random variable with mean $\mu_{f,t+1}$, covariance $k_{f,t+1}$, and standard deviation $\sigma_{f,t+1}$ given by
    \begin{align}
        \mu_{f,t+1}(\x) &=
        \bm{k}^T(\x)
        [\bm{k}(\bm{X}_t,\bm{X}_t) + \xi^2 I]^{-1} \bm{Y}_t \label{eq:gp_mean} \\
        %---------
        k_{f,t+1}(\x,\x') &= k_{f,t}(\x,\x'){-}\bm{k}^T(\x)
        [\bm{k}(\bm{X}_t,\bm{X}_t){+}\xi^2 I]^{-1} \bm{k}(\x') \nonumber \\
        %----------
        \sigma_{f,t+1}(\x) &= \sqrt{k_{f,t+1}(\x,\x)}. \label{eq:gp_std}
    \end{align}

    To incorporate data from multiple sources, we use an \mbox{AR-1} model, which models $f$ as a linear combination of a low-fidelity GP $f_L(\x)$ and an error GP $\delta(\x)$ according to
    \begin{align} \label{eq:ar1}
        f(\x) = \rho f_L(\x) + \delta(\x),
    \end{align}
    where $\rho$ is a scaling constant \cite{kennedy2000predicting}.
    In general, an AR-1 model is beneficial when the low-fidelity observations $\bm{X}_L$ are more abundant than the high-fidelity observations $\bm{X}_H$.

    Let $\bm{k}^{(L)}$ denote the kernel of $f_L(\x)$ and $\bm{k}^{(\delta)}$ denote the kernel of $\delta(\x)$. Then, letting $\bm{X} = [\bm{X}_L, \bm{X}_H]$, the covariance matrix of the AR-1 model has the form
    \begin{align} \label{eq:covariance_matrix}
        \bm{k}^{(MF)}(\bm{X},\bm{X})=
        \begin{bmatrix}
            \cov{L}{L}{L}  & \rho \cov{L}{H}{L} \\
            \rho \cov{H}{L}{L} & \rho^2 \cov{H}{H}{L} + \cov{H}{H}{\delta} \\
        \end{bmatrix},
    \end{align}
    where $\cov{L}{H}{L}$ is shorthand notation for the single-fidelity covariance matrix $\bm{k}^{(L)}(\bm{X}_L,\bm{X}_H)$.

    % From the covariance matrix, the posterior mean and variance at $\x_{S}$ are
    % \begin{align}
        % \mu_{t+1}(\x_{S}) &=
        % \begin{bmatrix}
            %     \rho \cov{S}{L}{L} & \rho^2 \cov{S}{H}{L}+ \cov{S}{H}{\delta}
            % \end{bmatrix} \nonumber\\
        % &\phantom{{}=}\times
        % \begin{bmatrix}
            %     \cov{L}{L}{L}  & \rho \cov{L}{H}{L} \\
            %     \rho \cov{H}{L}{L} & \rho^2 \cov{H}{H}{L} + \cov{H}{H}{\delta}
            % \end{bmatrix} ^{-1}
        % \begin{bmatrix}
            %     \bm{Y}_L \\
            %     \bm{Y}_H
            % \end{bmatrix}, \\
        % %
        % \sigma_{t+1}^2(\x_{S}) &=
        % \rho^2 \cov{S}{S}{L} + \cov{S}{S}{\delta} - \begin{bmatrix}
            %     \rho \cov{S}{L}{L} & \rho^2 \cov{S}{H}{L} + \cov{S}{H}{\delta}
            % \end{bmatrix} \nonumber \\
        % &\phantom{{}=}\times\begin{bmatrix}
            %     \cov{L}{L}{L}  & \rho \cov{L}{H}{L} \\
            %     \rho \cov{H}{L}{L} & \rho^2 \cov{H}{H}{L} + \cov{H}{H}{\delta}
            % \end{bmatrix}^{-1} \nonumber\\
        % &\phantom{{}=}\times
        % \begin{bmatrix}
            %     \rho \cov{L}{S}{L} \\
            %     \rho^2 \cov{H}{S}{L} + \cov{H}{S}{\delta}
            % \end{bmatrix}.
        % \end{align}

    \subsection{Problem Statement}\label{sec:problem}
    Consider a finite domain $\mathcal{X} \subset \mathbb{R}^n$, with $\x = (x_1, \dots, x_n) \in \mathcal{X}$. Let $f: \mathcal{X} \to \mathbb{R}$ be an unknown realization of a GP and let $\x^*$ be a minimizer of $f$. Given a safety barrier $h \in \mathbb{R}$ and precision $\epsilon>0$, our goal is to design a sequence $\{\x_t\}_{t\in \mathbb{N}}$ such that for some sufficiently large $t^*$,
    \[
    f(\x_t) < f(\x^*) + \epsilon,\; \forall t > t^*; \text{ and } f(\x_t)\leq h \; \forall t \in \mathbb{N}.
    \]

    We develop an iterative algorithm to design such a sequence $\{\x_t\}_{t \in \mathbb{N}}$. We apply this framework to the multi-fidelity case when an approximation of $f(\x)$ is available.

    \section{Algorithms and Main Results}
    In this section, we first review the \safeopt\ algorithm, which forms the framework of \safeslope. Next, we introduce \safeslope\ and describe how it deviates from \safeopt. We then discuss how \safeslope\ applies to MF-GPs, then discuss the theoretical properties of this algorithm.

    \subsection{The \safeopt\ Algorithm \cite{sui2015safe}}
    \safeopt\ is an exploration algorithm that uses the Lipschitz constant $L$ of a function $f$ to avoid searching in an unsafe domain. To accomplish this, \safeopt\ uses the predictive confidence interval
    \begin{align} \label{eq:Q}
        Q_{f,t}(\x) := \left[Q_{f,t}^-(\x),  Q_{f,t}^+(\x)\right],
    \end{align}
    where $Q_{f,t}^\pm(\x) := \mu_{f, t-1}(\x) \pm \beta_{f,t}^{1/2}\sigma_{f,t-1}(\x)$ and $\beta_{f,t}$ is a parameter which controls exploration.

    \textbf{Step 1:}
    Given an initial safe set $S_0$, we define $C_{f,0}(\x) := [h,\infty),$ $\forall \x \in S_0$ and $\mathbb{R}$ otherwise. Then, the nested confidence interval
    $C_{f,t}(\x) = C_{f,t-1}(\x) \cap Q_{f,t}(\x)$
    is used to define the upper and lower confidence bounds of $f$ as
    \begin{align} \label{eq:u_ell}
        u_{f,t}(\x):= \max C_{f,t}(\x) \text{ and }
        \ell_{f,t}(\x):= \min C_{f,t}(\x).
    \end{align}

    \textbf{Step 2:}
    These confidence bounds are used to establish the subsequent safe sets $S_t$ according to
    \begin{align*}
        S_t = \underset{\bm{x}\in S_{t-1}}{\bigcup} \left\{
        \bm{x}' \in \mathcal{X} \mid u_{f,t}(\bm{x}) + L d(\bm{x},\bm{x}') \leq h
        \right\},
    \end{align*}
    where $d(\bm{x},\bm{x}')$ is the distance between $\x$ and $\x'$.

    \textbf{Step 3:}
    Two subsets of $S_t$ guide the search process. The set of points that potentially minimize $f$ is given by
    \begin{align*}
        M_t = \left\{ x \in S_t \middle|\ell_{f,t}(\x) \leq \min_{\x' \in S_t} u_{f,t}(\x')  \right\}.
    \end{align*}

    \textbf{Step 4:}
    Meanwhile, the set of points that potentially increase the size of $S_t$ is given by
    \begin{align*}
        G_t = \left\{ x \in S_t \middle|g_t(\x) >0 \right\},
    \end{align*}
    where $g_t(\x)$ is the cardinality of the set of points that sampling at $\x$ could add to $S_t$, defined by
    \begin{align*}
        g_t(\bm{x}) := \Big|\left\{
        \bm{x}' \in \mathcal{X} \backslash S_t \mid \ell_{f,t}(\bm{x}) + L d(\bm{x},\bm{x}') \leq h
        \right\}\Big|.
    \end{align*}

    \textbf{Step 5:} From the union of $M_t$ and $G_t$, \safeopt\ selects points using the width of the confidence interval $w_t(\x) := u_{f,t}(\x) - \ell_{f,t}(\x)$ according to the function
    \begin{align} \label{eq:sampling}
        \x_t \in \underset{\x \in M_t \cup G_t}{\arg\max } \; w_t(\x).
    \end{align}

    \subsection{The \safeslope\ Algorithm}
    The \safeslope\ algorithm is an adaptation of \safeopt\ with the following modification: \emph{we assume the global Lipschitz constant is unknown} and instead use local slope predictions to avoid searching beyond the safety limit.
    \begin{algorithm} [t]
        \caption{\safeslope} \label{alg:SS}
        \begin{algorithmic}[1]
            \STATE \textbf{Input:} GP $f$, Safe limit $h$, Discrete grid domain $\mathcal{X}$, Initial safe set $S_0$, Grid incidence matrices $W_i$.
            \STATE $C_{f,0}(\x) \leftarrow [h,\infty), \hspace{5pt} \forall \x \in S_0$.
            \STATE $C_{f,0}(\x) \leftarrow \mathbb{R}, \hspace{5pt} \forall \x \in \mathcal{X} \backslash S_0$.
            \STATE $C_{m_i,0}(\x) \leftarrow \mathbb{R}, \hspace{5pt} \forall \x \in \mathcal{X}$ and for each $i=1,\dots, n$.
            \FOR {$t = 1,2,\dots$}
            \STATE Calculate $\bm{\mu}_{f,t}(\mathcal{X})$ using \eqref{eq:gp_mean}.
            \STATE Calculate $\bm{k}_{f,t}(\mathcal{X},\mathcal{X})$ using \eqref{eq:gp_std}.
            \FOR {$i=1,\cdots,n$}
            \STATE $\bm{\mu}_{m_i,t}(\mathcal{X}) \leftarrow W_i\cdot \bm{\mu}_{f,t}(\mathcal{X})$
            \STATE $\bm{k}_{m_i,t}(\mathcal{X},\mathcal{X}) \leftarrow W_i \cdot \bm{k}_{f,t}(\mathcal{X},\mathcal{X}) \cdot W_i^T$
            \STATE Compute $q_{m_i,t}(\x,\x'), \hspace{5pt} \forall \x \in \mathcal{X}$ using \eqref{eq:Q_m}.
            \STATE $\hat{u}_{m_i,t} \leftarrow \min\{q_{m_i,t}, \hat{u}_{m_i,t-1}\}$
            \ENDFOR
            \STATE Compute $Q_{f,t}(\x) \hspace{5pt} \forall \x \in S_{t-1}$
            \STATE $C_{f,t}(\x) \leftarrow C_{f,t-1}(\x) \cap Q_{f,t}(\x)$
            \STATE $S_t \leftarrow \cup_{\x\in S_{t-1}} \cup_{i=1,\dots,n} \left\{\x' {\in} V_i(\x) \mid s_t(\x,\x')\leq h \right\}$ \label{alg_line:S_t}
            \STATE $M_t \leftarrow \left\{ \x \in S_t \mid\ell_t(\x) \leq \min_{\x' \in S_t} u_t(\x')  \right\}$ \label{alg_line:M_t}
            \STATE $G_t \leftarrow \left\{ \x \in S_t \mid g_t(\x) >0 \right\}$ \label{alg_line:G_t}
            \STATE $\x_t \leftarrow \underset{\x \in M_t \cup G_t}{\arg\max}w_t(\x)$
            \STATE $y_t \leftarrow f(\x_t) + \xi_t $
            \ENDFOR
        \end{algorithmic}
    \end{algorithm}

    To do so, we model the slopes of $f$ as GPs. For ease of presentation, we organize $\mathcal{X}$ into a hypercube with $r^n$ points. Along each axis $i \in \{1,\dots,n\}$, we create an incidence matrix $W_i$ with size ${(r-1)r^{n-1} \times r^n}$. Each $W_i$ corresponds to the union of directed line graphs along the $i$-th axis.
    Then, at iteration $t$, we represent the slopes between adjacent points along the $i$-th axis using $m_i \in \mathbb{R}^{(r-1)r^{n-1}}$. Each $m_i$ is a realization of a GP with mean and covariance
    \begin{align*}
        \bm{\mu}_{m_i} = W_i \cdot \bm{\mu}_{f,t}(\mathcal{X}),\quad
        \bm{k}_{m_i} = W_i \cdot \bm{k}_{f,t}(\mathcal{X}) \cdot W_i^T.
    \end{align*}
    Essentially, the elements of $m_i$ consist of evaluations of
    \begin{align*}
        m_i(\x,\x') = [\mu_f(\x')-\mu_f(\x)] / d(\x',\x),
    \end{align*}
    where $\x$ and $\x'$ are adjacent points along the $i$-th axis, $x_i' > x_i$, and $d(\x',\x)$ is the distance between $\x'$ and $\x$.

    \textbf{Step 1:}
    We preserve the format of \safeopt's safety condition by using the magnitude of the slope. Here, we use the greatest magnitude of the confidence bounds, defined by
    \begin{align} \label{eq:Q_m}
        q_{m_i,t}(\x,\x') {:=} \max \left\{ \text{abs}(Q_{m_i,t}^-(\x,\x')),  \text{abs}(Q_{m_i,t}^+(\x,\x'))\right\},
    \end{align}
    where
    \[
    Q_{m_i,t}^\pm(\x,\x') := \mu_{m_i, t-1}(\x,\x') \pm \beta_{m,t}^{1/2}\sigma_{m_i,t-1}(\x,\x').
    \]
    Then, we replace $L$ with the nested upper bound on the slope
    \begin{align} \label{eq:u_hat}
        \hat{u}_{m_i,t}(\x,\x') := \min\{q_{m_i,t}(\x,\x'), \hat{u}_{m_i,t-1}(\x,\x')\},
    \end{align}
    where $\hat{u}_{m_i,0}=\infty$.\footnote{Instead of using the confidence bound with the greatest magnitude $\hat{u}_{m_i,t}$, the upper bound of each $m_i$ (i.e., $Q_{m_i,t}^+$) is used instead. In this case, displacement is used instead of distance. However, in our numerical simulations, we found this upper bound to be inferior to $\hat{u}_{m_i,t}$.}

    \textbf{Step 2:}
    We now redefine the safe set as
    % (Alg. \ref{alg:SS}, line \ref{alg_line:S_t}) as
    \begin{align}\label{eq:S_t}
        S_t =\bigcup_{\x\in S_{t-1}} \bigcup_{i=1,\dots,n} &\big\{\x' \in V_i(\x) \mid s_t(\x,\x') \leq h \big\},
    \end{align}
    where
    \begin{align*}
        s_t(\x,\x') = u_{f,t}(\x) + \hat{u}_{m_i,t}(\x,\x')\cdot d(\x,\x')
    \end{align*}
    and the vicinity $V_i$ of $\x$ is given by
    \begin{align*}
        V_i(\x) = \left\{ \x' \in \mathcal{X} \middle| \x', \x \text{ are adjacent and } x_i'= x_i \right\}.
    \end{align*}

    \textbf{Steps 3 and 4:}
    The definitions of $M_t$ and $G_t$ are the same as those in SAFE-OPT, but the growth criterion becomes
    \begin{align*}
        g_t(\x) = \Big|\left\{
        \x' \in V_i(\x) \backslash S_t \middle|\ell_{f,t}(\x) + \hat{u}_{m_i,t} d(\x,\x') \leq h
        \right\}\Big|.
    \end{align*}

    \textbf{Step 5:}
    Similar to \safeopt, points are sampled using the redefined $M_t$ and $G_t$ according to \eqref{eq:sampling}.

    \subsection{Multi-fidelity Extension of \safeslope}
    We can use \safeslope\ to sample points from the highest fidelity of an MF-GP. Consider an AR-1 GP with  fidelities, $f_L$ and $f$. We evaluate $f_L$ at every $\x \in \mathcal{X}$ to construct a data set $(\bm{Y}_L,\bm{X}_L)$.
    We also evaluate $f$ at a starting point $\x_0 = \underset{\x\in\mathcal{X}}{\arg\min}f_L(\x)$.
    Then, with $\x_0$ as $S_0$, \safeslope\ is used to explore the AR-1 GP and find $\x^*$.
    This extension is formalized in Algorithm \ref{alg:MF_SS}.

    \begin{algorithm} [thb]
        \caption{Multi-Fidelity \safeslope\ Optimization} \label{alg:MF_SS}
        \begin{algorithmic}[1]
            \STATE \textbf{Input:} Safe Limit $h$, Discrete domain $\mathcal{X}$
            \STATE Assume $f = \rho f_L + \delta$
            \STATE Evaluate $f_L(\x)$ for all $\mathcal{X}$
            \STATE $\x_0 \leftarrow \underset{\x\in\mathcal{X}}{\arg\min}f_L(\x)$
            \STATE Evaluate $f(\x)$  for $\x_0$
            \STATE $S_0 \leftarrow \x_0$
            \STATE Conduct \safeslope($f|(\bm{Y}_L,\bm{X}_L)$, $h$, $\mathcal{X}$, $S_0$)
        \end{algorithmic}
    \end{algorithm}

    \subsection{Reachability}
    Similar to \safeopt, the theoretical guarantees of \safeslope\ rely on the reachability operator. Define $\hat{u}_t := [\hat{u}_{m_1,0}, \dots, \hat{u}_{m_n,0}]^T$. Then the reachability operator at time $t$ is the set of points given by
    \begin{align*}
        &R_{\epsilon,\hat{u}_t}(S) := \\
        & S \cup\left\{\x' \in \mathcal{X} \middle|
        \begin{array}{c}
            \exists \x \in S, \exists i \in \{1,\dots,n\}, \x' \in V_i(\x),\\
            f(\x)  + \hat{u}_{m_i,t}(\x,\x') \cdot d(\x,\x') + \epsilon \leq h
        \end{array}
        \right\},
    \end{align*}
    where $\hat{u}_{m_i,t}(\x,\x')$ is the upper bound on the slope between $\x$ and $\x'$ at time $t$.
    Given the current set of safe points,
    the reachability operator provides the total collection of points that could be sampled as $f$ is learned within $S$.

    The $T$-step reachability operator is defined by
    \begin{align}
        R_{\epsilon}^T(S) := R_{\epsilon,\hat{u}_T}(R_{\epsilon,\hat{u}_{T-1}}\dots(R_{\epsilon,\hat{u}_0}(S))).
    \end{align}
    By taking the limit, we obtain the closure set $\bar{R}_{\epsilon}(S) := \lim_{T\rightarrow \infty} R_{\epsilon}^T(S)$.
    Note: In \safeslope, we restrict the expansion of the safe set to the vicinity of the previous safe set. This restriction does not affect the closure of the reachability set, but only slows down the rate of expansion.
    Because \safeslope\ never explores outside $\bar{R}_\epsilon(S_0)$ with probability 1, we modify our optimization goal from Section~\ref{sec:problem} to take the equivalent form,
    \begin{align*}
        f^*_\epsilon = \min_{\x\in\bar{R}_{\epsilon}(S_0)} f(\x).
    \end{align*}

    \subsection{Theoretical Results}
    For Bayesian approaches, we measure the information gain after sampling a set of points $A \subseteq \mathcal{X}$ as
    $I(\bm{y}_A; \bm{f}_A) = H(\bm{y}_A) - H(\bm{y}_A|f)$,
    where $\bm{y}_A$ is a random vector of noisy observations of $f$ evaluated at every point in $A$, $\bm{f}_A$ is the vector of true values of $f$ at every point in $A$, and $H$ is the entropy of the vector.
    The maximum information gain after $T$ evaluations of $f$ is given by
    \begin{align} \label{eq:gamma}
        \gamma_T = \underset{A\subset \mathcal{X}, |A|=T}{\max}I(\bm{y}_A; \bm{f}_A).
    \end{align}
    A bound on the $\gamma_T$ can be found in \cite[Eq. (8)]{srinivas2012information}.
    With the information gain defined, we now move to the main theorem.
    \smallbreak
    \begin{theorem}[Single-Fidelity \safeslope\ Guarantees] \label{thm:ss}
        Define $\hat{\x}_t := \arg \min_{x\in S_t} u_{f,t}(\x)$. Select $\delta_f, \delta_m \in (0,1)$. Set $\beta_{f,t}= 2 \log(|\mathcal{X}|\pi_t/\delta_f)$ and $\beta_{m,t}= 2 \log(|\mathcal{X}|n\pi_t/\delta_m)$, where $\sum_{t\geq1} \pi_t^{-1}=1$ with $\pi_t>0$. Given an initial safe set $S_0 \neq \varnothing$, with $f(\x) \leq h$ for each $\x \in S_0$, let $t^*$ be the smallest positive integer satisfying
        \begin{align*}
            \frac{t^*}{\gamma_{t^*}\beta_{f,t^*}}\geq \frac{C_1\left(|\Bar{R}_{0}(S_0)|+1 \right)}{\epsilon^2},
        \end{align*}
        where $C_1 = 8 v^2 /\log(1+v^2\xi^{-2})$, $v^2$ is the kernel variance, and $|\bullet|$ denotes cardinality. Then, for any $\epsilon > 0$, using \safeslope\ with $\beta_{f,t}$ and $\beta_{m,t}$ results in the following.
        \begin{compactitem}
            \item With probability at least $1-\delta_f-\delta_m$,
            \[
            \forall t \geq 1, f(\x_t)\leq h.
            \]
            \item With probability at least $1-\delta_f$,
            \[
            \forall t \geq t^*, f(\hat{\x}_t) < f^*_\epsilon + \epsilon. \tag*{\qedsymbol}
            \]
        \end{compactitem}
    \end{theorem}
    \smallskip

    The first point of Theorem \ref{thm:ss} states that with high probability, \safeslope\ will sample points under a threshold $h$. This probability is directly tied to $\beta_f$ and $\beta_m$, parameters that quantify the algorithm's tendency to explore points in unexplored regions. The second point states that with high probability, after time $t^*$, the minimum yielded by \safeslope\ will fall within an $\epsilon$-neighborhood of $f_\epsilon^*$. This value of $t^*$ scales intuitively with the information gain $\gamma_{t^*}$, since more information to learn requires a greater search iteration count. Because $\gamma_{t^*}$ lacks a closed-form solution, a bound on $\gamma_{t^*}$ is typically used instead.

    % \begin{remark}
        %     Comparing this result with Theorem \ref{thm:ss} of \cite{sui2015safe} reveals that replacing $L$ with $\hat{u}_{m_i,t}$ affects only the probability of safety and not the probability related to $t^*$.
        % \end{remark}

    % Since the maximum information gain quantity is difficult to evaluate, the following upper bound on the maximum information gain is often used for analysis.
    % \begin{proposition} \label{prop:gamma_bound}
        % The maximum information gain of a GP evaluated over a set of points $\bm{X}_T$ is upper bounded by
        % \begin{align*}
            %     \gamma_T \leq \frac{1/2}{1-e^{-1}}
            %     \max_{m_1,...,m_T} \sum_{t=1}^T \log \left(1+\xi^{-2} m_t \lambda_t \right),
            % \end{align*}
        % where $\lambda_t$ are the eigenvalues of the matrix $\bm{k}(\bm{X}_t, \bm{X}_t)$. \hfill \qedsymbol %\cite{srinivas2012information}
        % \end{proposition}

    Our second main result is an extension of Theorem \ref{thm:ss} to an AR-1 model. But first, we establish an upper bound on the information gain $\gamma_T$ for an AR-1 model.

    \medskip

    \begin{theorem}[Information Gain Bound for an AR-1 GP]\label{thm:info_gain}
        Consider the information gain $\gamma_T$ from \eqref{eq:gamma}.
        For a linear auto-regressive GP with noise-free ($\xi_L^2=0$) low-fidelity observations at $\bm{X}_L$ and high-fidelity observations at $\bm{X}_H \subseteq \bm{X}_L$, the information gain $\gamma_T$ is upper bounded by
        \begin{align} \label{eq:gamma_tilde}
            %\gamma_T \leq
            \tilde{\gamma}_T := \frac{1/2}{1-e^{-1}}
            \max_{m_1,...,m_T} \sum_{t=1}^T \log \left(1+\xi^{-2} m_t \lambda_t^{(\delta)} \right),
        \end{align}
        where $\sum_{i=1}^T m_i=T$ and $\lambda_t^{(\delta)}$ are the eigenvalues of the error covariance matrix $\cov{H}{H}{\delta}$.
    \end{theorem}

    \begin{proof}
        Suppose we have the high- and low-fidelity input points $\bm{X}_H$ and $\bm{X}_L$, where $\bm{X}_H \subseteq \bm{X}_L$,  $\bm{X}_{H'} = \bm{X}_L\backslash \bm{X}_H$, and each entry of $\bm{X}_L$ is unique. Then, $\bm{X}_L = \bm{X}_H \cup \bm{X}_{H'}$. Since the covariance matrix is always positive definite, $\cov{L}{L}{L}$ is invertible, and the covariance of the high-fidelity data conditioned on the low-fidelity data is given by
        \begin{align*}
            &\bm{k}(f_H(\bm{X}_H),f_H(\bm{X}_H) |f_L(\bm{X}_L)=\bm{y}_L,f_H(\bm{X}_H)=\bm{y}_H)\\
            &=
            \rho^2 \cov{H}{H}{L} + \cov{H}{H}{\delta} - \rho^2 \cov{H}{L}{L} [\cov{L}{L}{L}]^{-1} \cov{L}{H}{L}\\
            &=
            \rho^2 \cov{H}{H}{L} + \cov{H}{H}{\delta} - \rho^2 \begin{bmatrix}
                \cov{H}{H'}{L} & \cov{H}{H}{L}
            \end{bmatrix} \\
            &\phantom{{}=} \times
            \begin{bmatrix}
                \cov{H'}{H'}{L} & \cov{H'}{H}{L}\\
                \cov{H}{H'}{L} & \cov{H}{H}{L}
            \end{bmatrix} ^ {-1}
            \begin{bmatrix}
                \cov{H'}{H}{L} \\
                \cov{H}{H}{L}
            \end{bmatrix}\\% = \cov{H}{H}{\delta},
            &=
            \rho^2 \cov{H}{H}{L} + \cov{H}{H}{\delta} \\
            &\phantom{{}=}- \rho^2
            \left[
            (\cov{H}{H'}{L} - \cov{H}{H}{L} [\cov{H}{H}{L}]^{-1} \cov{H}{H'}{L}) \right.\\
            &\phantom{{}=-\rho^2[+}\times
            (\cov{H'}{H'}{L} - \cov{H'}{H}{L} [\cov{H}{H}{L}]^{-1} \cov{H}{H'}{L})\\
            &\phantom{{}=-\rho^2[+}\times
            (\cov{H'}{H}{L} - \cov{H'}{H}{L} [\cov{H}{H}{L}]^{-1} \cov{H}{H}{L}) \\
            &\phantom{{}=-\rho^2[}+ \left.(\cov{H}{H}{L} [\cov{H}{H}{L}]^{-1} \cov{H}{H}{L})
            \right]\\
            &= \cov{H}{H}{\delta},
        \end{align*}
        where the second to last line is obtained using properties of block matrix inversion.
        In words, the conditional covariance is simply the covariance of the error GP $\delta(\x)$.
        By applying the above result to \cite[Eq. (8)]{srinivas2012information}, we complete the proof.
    \end{proof}
    \begin{remark}
        As the quality of a low-fidelity model improves,
        the variance of the error GP approaches 0. Since the eigenvalues of a covariance matrix are directly proportional to the kernel's variance hyper-parameter, Theorem \ref{thm:info_gain} shows that improving the low-fidelity quality decreases the eigenvalues of $\cov{H}{H}{\delta}$, thereby decreasing the information gain.
    \end{remark}
    \medskip
    \begin{remark}
        Using an AR-1 model, a series of $p$ fidelities may be nested to obtain
        \begin{align} \label{eq:ar1_nested}
            f_p(\x) = \rho_p (f_{p-1}(\x)) + \delta_p(\x).
        \end{align}
        From the proof of Theorem \ref{thm:info_gain}, we see that the conditional covariance of a nested AR-1 model depends only on the highest level error GP $\delta_p(\x)$.
    \end{remark}
    \begin{theorem}[Multi-Fidelity \safeslope\ Guarantees] \label{thm:mf_ss}
        Assume $f$ is an AR-1 GP with the structure given in \eqref{eq:ar1}.
        Consider $\hat{\x}_t$, $\delta_f$, $\delta_m$, $\beta_{f,t}$, $\beta_{m,t}$, $\pi_t$, and $S_0$ as defined in Theorem \ref{thm:ss}.
        Let $t_{MF}^*$ denote the smallest positive integer satisfying
        \begin{align*}
            \frac{t_{MF}^*}{\tilde{\gamma}_{t_{MF}^*} \beta_{f,t_{MF}^*} }\geq \frac{C_1\left(|\Bar{R}_{0}(S_0)|+1 \right)}{\epsilon^2},
        \end{align*}
        where $\tilde{\gamma}_{t_{MF}^*}$ is defined by \eqref{eq:gamma_tilde}, $C_1 = 8 v_{MF}^2 /\log(1+v_{MF}^2\xi^{-2})$, and $v_{MF}^2$ is the variance of the AR-1 GP, given by $v_{MF}^2 = \rho v_L^2 + v_\delta^2$.
        Then, for any $\epsilon > 0$, using \safeslope\ with $\beta_{f,t}$ and $\beta_{m,t}$, with probability at least $1-\delta_f$,
        \[
        \forall t \geq t_{MF}^*, f(\hat{\x}_t) < f^*_\epsilon + \epsilon. \tag*{\qedsymbol}
        \]
    \end{theorem}

    This theorem indicates that the quality of a multi-fidelity model impacts the time $t_{MF}^*$ to identify an optimal $\hat{\x}$. In particular, improving the quality of the low-fidelity model lowers the information gain bound $\tilde{\gamma}_{t_{MF}^*}$, thereby decreasing the time to find an optimal $\hat{\x}$.

    \section{Numerical Results}
    We now apply \safeslope\ to our motivating scenario, in which we try to find the best controller for a system when an approximate model of the system exists.
    \smallbreak

    For the motivating scenario from Section~\ref{sec:scenario}, consider a $2\times 2$ LTI system. For the true system, we let
    \begin{align} \label{eq:AB_True}
        A = \begin{bmatrix}
            0.785 & -0.260 \\ -0.260 & 0.315
        \end{bmatrix},
        \hspace{5pt}
        B = \begin{bmatrix}
            1.475 \\ 0.607
        \end{bmatrix}.
    \end{align}

    By applying system identification \cite{ho1966effective} to \eqref{eq:AB_True} with $N_s=12$ snapshots, we obtain the approximate model,
    \begin{align} \label{eq:AB_Hat}
        \hat{A} = \begin{bmatrix}
            0.700 & -0.306 \\ -0.306 & 0.342
        \end{bmatrix},
        \hspace{5pt}
        \hat{B} = \begin{bmatrix}
            1.543 \\ 0.524
        \end{bmatrix}.
    \end{align}

    Since unstable controllers result in extremely large costs, we modify the cost functions to be
    \begin{align} \label{eq:f}
        f(\x) = \log(J(\x)), \hspace{5pt} f_{L}(\x) = \log(\hat{J}(\x)),
    \end{align}
    where $J$ and $\hat{J}$ are approximated by a 20-step horizon quadratic cost with $Q=I$, $R=1$ and $\x$ now represents the choice of controller gains.  Gaussian noise with variance $\xi^2  = 10^{-4}$ and $\xi_L^2  = 10^{-8}$ is added to evaluations of $f$ and $f_L$ to ensure kernel matrices are well-conditioned.

    \begin{figure}[b]
        \centering
        \includegraphics[width = 0.48\textwidth]{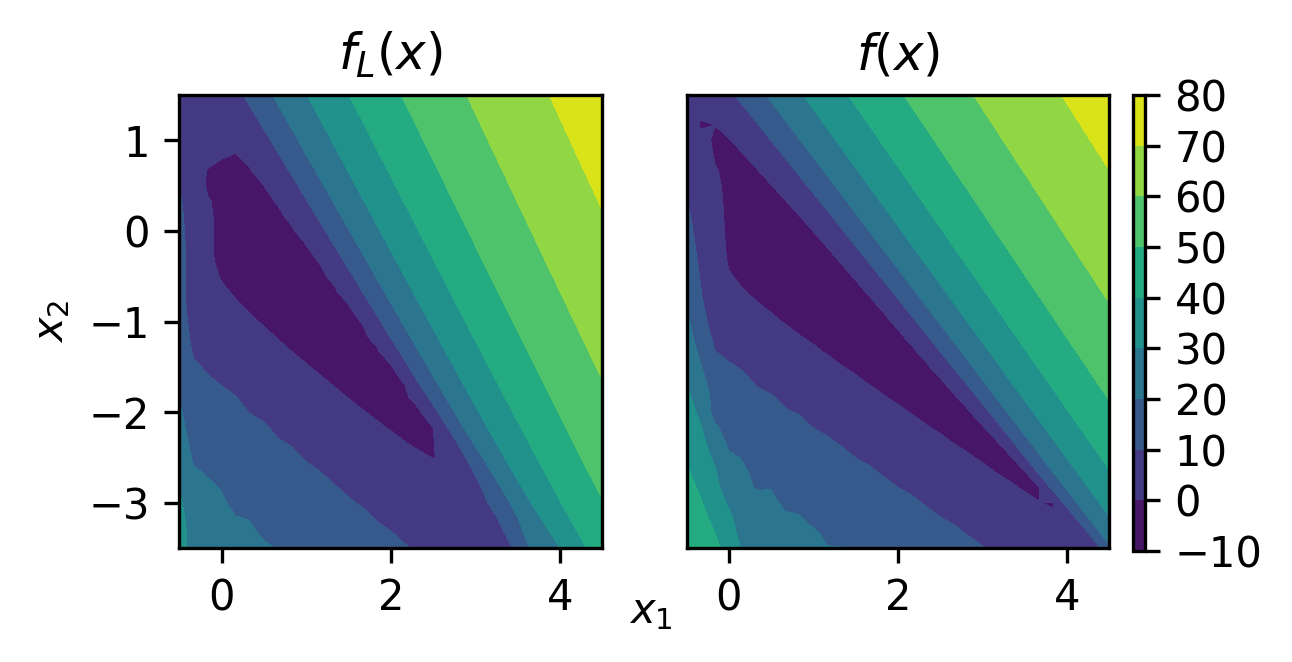}
        \caption{True Plots of $f_L(\x)$, $f(\x)$. Darker regions indicate lower LQR costs}
        \label{fig:true_plot}
    \end{figure}

    Our goal is to find the controller gains $\x^* = [x_1^* \hspace{5pt} x_2^*]$ such that \eqref{eq:f} is minimized. First, we set a search domain $\mathcal{X}$ and select an initial safe set $S_0$. In practice, input constraints and low-fidelity data could guide the choice of $\mathcal{X}$ and $S_0$. Here, we set $x_1 \in [-0.5,4.5]$, $x_2 \in [-3.5,1.5]$, and resolution $r=26$. Mat\'ern kernels are used to correlate points for each fidelity \cite{rasmussen2006gaussian}.
    For 10 different $S_0$'s of three points each, we observe the safety and regret of \safeslope\ with parameters $h=0$, $\delta_f = 0.1$, $\delta_m = 0.1$, and $\pi_t = t^2 \pi^2/6$.
    We compare \safeslope\ to \safeucb, a naive approach that solely relies on $u_{f,t}(\x)$ for safety and selects points according to
    \begin{align*}
        \x_t = \underset{\x \in S_t}{\arg\max} \hspace{5pt} w(\x_t), \text{ where } S_t = \left\{\x \in \mathcal{X} \middle| u_{f,t}(x) \leq h\right\}.
    \end{align*}
    We use \safeucb\ with $h=0$, $\delta_f = 0.1$, and $\pi_t = t^2 \pi^2/6$.

    Numerically, we achieve better results when the definitions of the confidence bounds are relaxed to follow \eqref{eq:Q} and \eqref{eq:Q_m} rather than their nested counterparts. As such, the following results are obtained using the unnested confidence bounds of $f$ and $m_i$.

    The search progressions of SAFE-SLOPE for the single- and multi-fidelity models is displayed in Fig \ref{fig:safeslope_progression}. Compared to the single-fidelity search, the multi-fidelity search samples fewer points above the safety threshold $h$.
    \begin{figure}
        \centering
        \includegraphics[width = 0.48\textwidth]{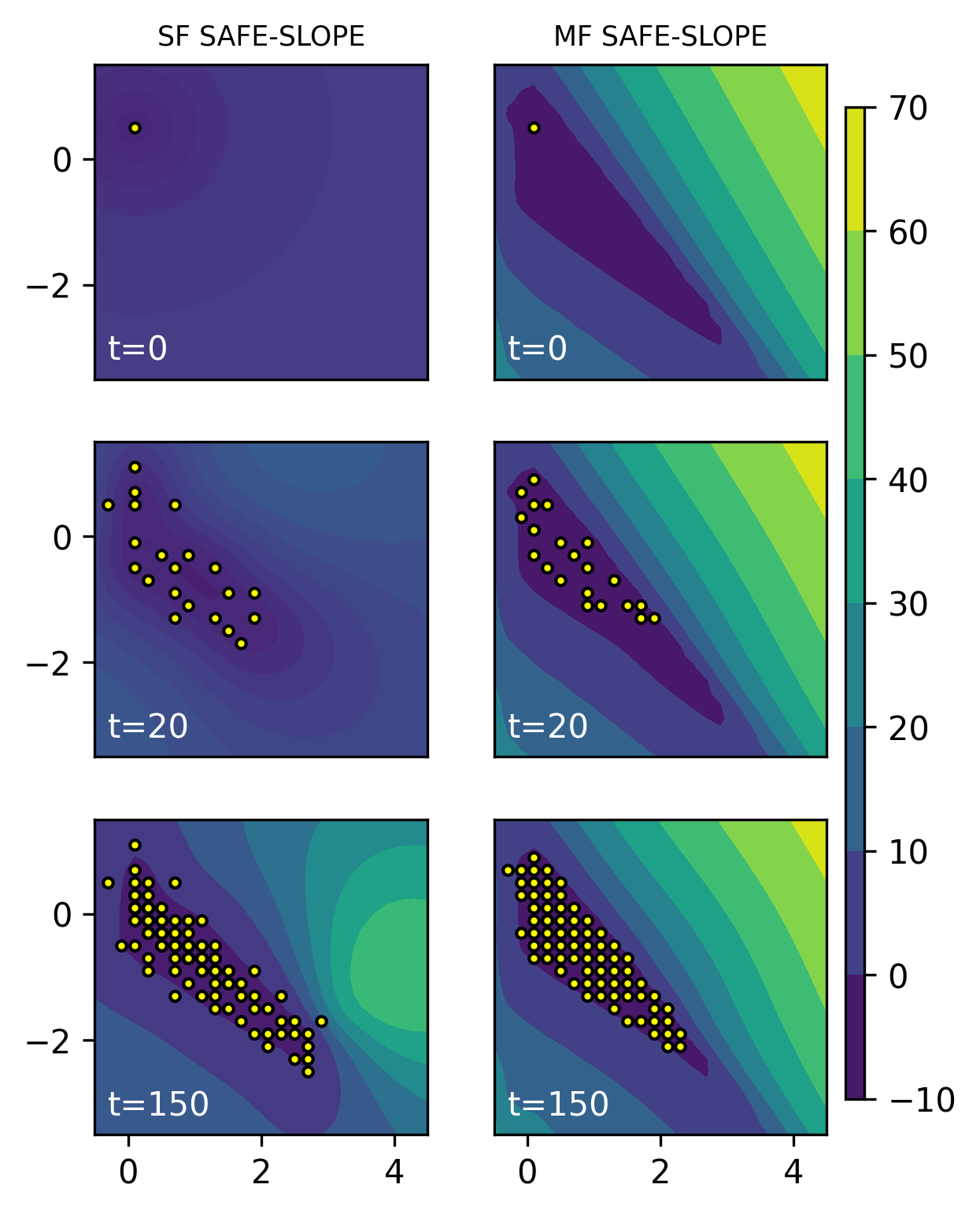}
        \caption{\small{Progression of sampling using SAFE-SLOPE on single- and multi-fidelity models. Darker regions indicate areas of lower LQR costs. By utilizing an AR-1 model, fewer unstable points are tested.}}
        \label{fig:safeslope_progression}
    \end{figure}

    To compare \safeslope\ to \safeucb, we use the cumulative regret up to time $T$, given by $R_T = \sum_{t=0}^T \left(f(\x_t)-f^* \right)$.
    Fig. \ref{fig:cumulative_regret} plots the cumulative regret and cumulative number of unsafe samples over 150 iterations. We see that in this example the multi-fidelity \safeslope\ algorithm performs the best, with a plateau in regret after 25 iterations. In general, \safeslope\ obtains better cumulative regret than \safeucb\ at higher iteration counts. By limiting evaluations to growth or minimizer points, \safeslope\ eliminates non-ideal points in fewer trials. This differs from \safeucb, which seeks to limit uncertainty across all safe points, rather than growth and minimizer points only.
    We also see both algorithms sample fewer unsafe points on MF models, with MF \safeslope\ sampling the fewest unsafe points on average.

    \begin{figure}
        \centering
        \includegraphics[width = 0.4\textwidth]{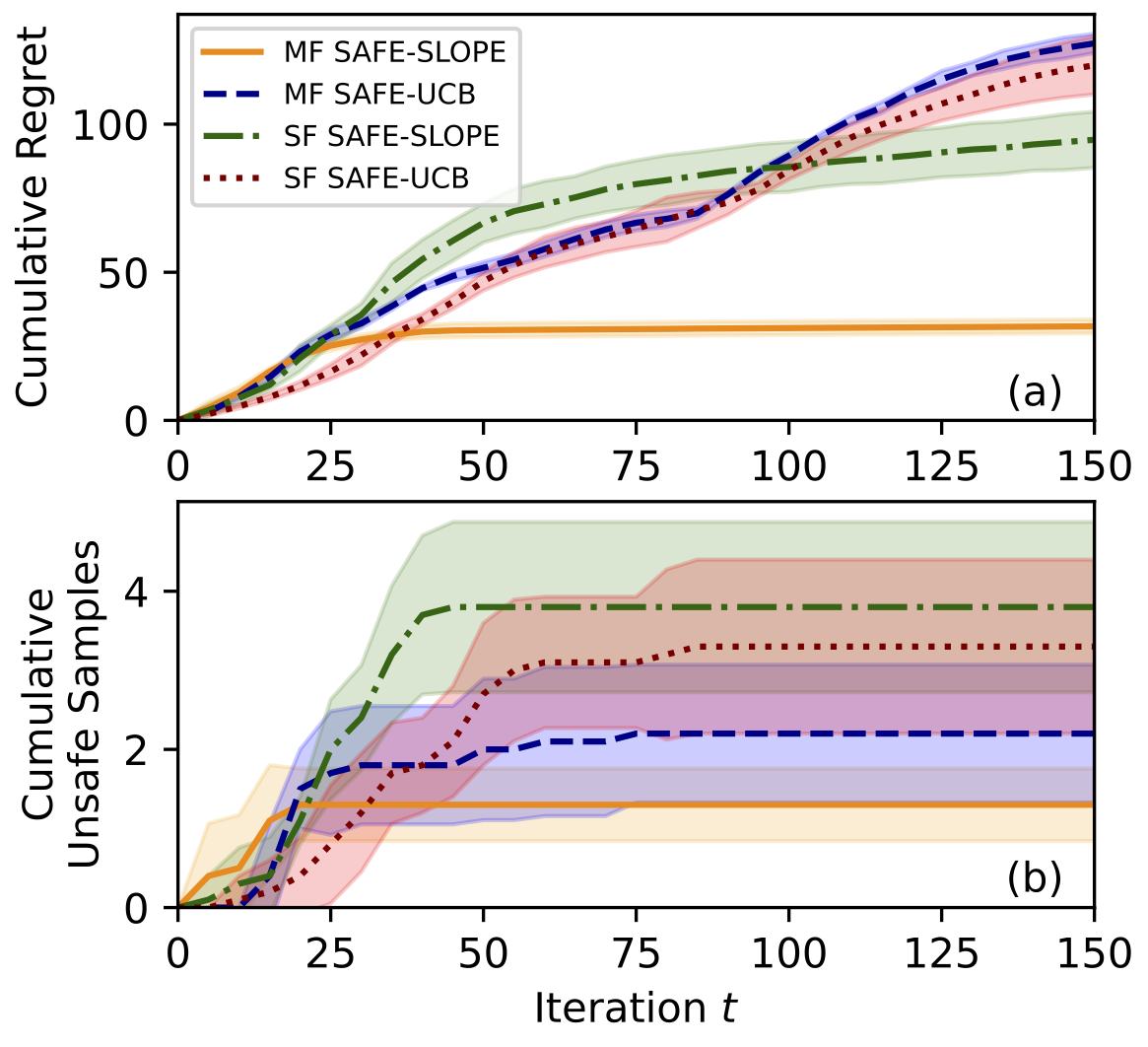}
        \caption{\small{(a) Cumulative regret and (b) the cumulative number of unsafe samples using \safeslope\ and \safeucb, averaged across 10 trials. Error bars indicate one standard deviation.}}
        \label{fig:cumulative_regret}
        \vspace{-16pt}
    \end{figure}

    \section{Discussion and Extensions}
    \subsection{Extension to Continuous Space}
    The presented form of \safeslope\ algorithm is limited to operating on  discretized spaces.  Here, we provide an intuition of how the algorithm may be extended to continuous spaces.

    Instead of relying on slopes, derivative estimation could be used for each point. The derivative of a GP is another GP, and a joint GP can be written to describe both the function and its derivative \cite{rasmussen2006gaussian}. This allows for the calculation of a posterior distribution of the derivatives conditioned on the function values, which could then be used to compute upper confidence bounds on the derivatives.

    The vicinity operator $V(\x)$ could be adjusted to return an $\epsilon$-neighborhood around $\x$. Either the derivative estimate at $\x$ or the maximum derivative estimate in the $\epsilon$-neighborhood could be used for safety.

    While the \safeslope\ algorithm would still be well-defined under these additional assumptions, the main challenge would be to appropriately define the $\beta$ terms and estimate the failure probabilities.

    By following the lines of analysis in the proofs of Theorems 2 and 3 from \cite{srinivas2012information}, an analysis could be conducted to branch this approach to the continuous space given bounds on the derivatives of $f$ or $f$ being a sample from an Reproducing Kernel Hilbert Space. By completing this analysis, a form of the first point of Theorem 3.1 would be extended to continuous spaces.
    This analysis is expected to be significantly more involved than the present discrete space counterpart, which happens to be more intuitive at the expense of a small convergence due to the discretization.

    In fact, the effects of discretization can be bounded by examining the point $\x^\dagger$ nearest point to $\x^*$. Specifically, if $\mathcal{X}$ is a uniform $n$-dimensional hypercubic grid with $r^n$ points and spacing $\Delta x$ between points in a given direction, then the distance $d(\x^\dagger,\x^*)$ is upper-bounded by $\Delta x \sqrt{n}/2$.
    Given a Lipschitz constant $L$, we bound
    \[
    f(\x^\dagger)-f(\x^*) \leq d(\x^\dagger,\x^*) \cdot L \leq  \frac{\Delta x \sqrt{n}}{2}L := \epsilon^\dagger.
    \]

    Since both \safeopt\ and \safeslope\ rely on the cardinality $|\Bar{R}_{0}(S_0)|$ for the convergence time $t^*$ guarantees, the second point of Theorem \ref{thm:ss} would not easily translate to continuous spaces.
    However, a fundamental limitation of using GPs is that their practical application requires discretization either in formulation or in implementation. This challenge is not unique to our algorithm, but applies to all work involving GPs. Setups formulated in continuous space typically resort to random sampling in order to determine extrema \cite{nowak2005relaxation,zabinsky2009random}. As a result, the discretized formulation of the presented algorithm remains a viable option for sampling GPs.

    \subsection{Application to Disturbance Models}
    In this paper, our use of the linear auto-regressive approach is motivated by a scenario in which we possess a true unknown system and a close approximation of it. The AR-1 model may also be applied to modeling disturbances on LTI systems.
    Here, we describe two major classes of disturbances for which our method could be applied.

    \begin{enumerate}
        \item \textbf{Linear Disturbance}: Consider an unknown deterministic disturbance that is linear in $\bm{z}$ and $u$. This typically arises when modeling drag forces on a robot caused by its environment. Then there exists an $A_d, B_d$ such that
        \[
        \bm{z}_{j+1} = [\hat{A} + A_d]\bm{z}_j + [\hat{B} + B_d]u.
        \]
        In this case, our model directly applies to this setup. The low-fidelity model corresponds to the disturbance-free dynamics $(\hat{A},\hat{B})$ while the high-fidelity model corresponds to the disturbance-impacted dynamics $(A,B)=(\hat{A} + A_d, \hat{B} + B_d)$. We assume $(\hat{A},\hat{B})$ to be known through some type of modeling or perfect-environment testing while the true system $(A,B)$ is a black-box.

        \item \textbf{Additive Disturbance}: Consider a known linear time-invariant system $(A,B)$ with additive disturbance $\bm{d}$ and the evolution
        \[
        \bm{z}_{j+1} = A\bm{z}_j + Bu_j + \bm{d}_j.
        \]
        Assuming perfect state feedback, the past values of the disturbance can be computed at all times. Thus, we consider control inputs of the form $u_j = -(K_H \bm{z}_j + K_d \bm{d}_{j-1})$ with the goal of finding a $K^* = [K_H^*\ K_d^*]$ which minimizes the quadratic cost of the system.
        We examine two possible configurations for how the AR-1 model may handle an additive disturbance.

        \begin{enumerate}
            \item \textbf{Known Additive Disturbance Model/Evolution}:
            In this case, we assume the evolution of $\bm{d}$ is known. Our model directly extends to this setup with an increase in the dimension of the search space. (We now search for $K_d$ in addition to $K_H$.) The low- and high- fidelities model the LQR cost of the disturbance-impacted systems as the gains $K$ change. As before, the low-fidelity model corresponds to an approximation $(\hat{A},\hat{B})$ while the high-fidelity model corresponds to the true system $(A,B)$.

            \item \textbf{Unknown Additive Disturbance}:
            In this case, we assume that $(A,B)$ is known but the disturbance $\bm{d}$ is unknown (but deterministic).
            Here, the low-fidelity $f_L(K)$ models the quadratic cost of the disturbance-free $(A,B)$. As such, $f_L$ is independent of $K_d$. The high fidelity model takes the form
            \[
            f_H(K) = \bm{\rho} f_L(K) + \delta(K),
            \]
            where $\delta$ accounts for the effects of disturbance and the control gain $K_d$. Essentially, the low-fidelity model acts as a prior in the directions of $K_H$ but does not inform the GP in the directions of $K_d$.

            The error GP $\delta(K)$ is well-suited to model the additive disturbance. By substituting our choice of $u$ into the system, we obtain
            \[
            \bm{z}_{j+1} = (A-B K_H)\bm{z}_j + (I-B K_d)\bm{d}_j.
            \]
            For a quadratic cost, the total cost of the disturbed system will be the sum of the cost of the undisturbed system plus an additional term contributed by the disturbance.
        \end{enumerate}
        Note, for the additive disturbance, if $\bm{d}$ does not diminish to 0, a discount factor may need to be applied the cost functions in order to prevent an infinite cost. Alternatively, one may assume a finite-energy disturbance with bounded $\ell_2$-norm.
    \end{enumerate}

    \section{Conclusion}
    We propose \safeslope, a safe exploration algorithm that leverages a function's posterior mean to predict its slopes. We preserve the safety result from \safeopt\ with a reduction in probability. By applying \safeslope\ to an AR-1 GP, we show the search time for an optimal point corresponds to the quality of the low-fidelity approximation. Finally, we examine \safeslope's performance by comparing it to a naive approach applied to single- and multi-fidelity models. We observe that applying \safeslope\ to an MF-GP achieves lower cumulative regret while sampling fewer unsafe points.

    Future research includes applying \safeslope\ to nonlinear systems, LTI systems with disturbances, or experimental robotic applications.
    Another direction is designing a search algorithm which can select either fidelity for evaluation.

    \bibliographystyle{IEEEtran}
    \bibliography{safeslope}

    \appendix%[Proof of Theorem 3.1]
    The following steps compose the proof of Theorem \ref{thm:ss}.
    We start by restating the upper confidence bound from Lemma 5.1 in \cite{srinivas2012information}.
    \begin{lemma}[UCB Bound] \label{lemma:ucb}
        Let $f$ be a function sampled from a GP. For all $t \geq 1$ and $\beta_{f,t} = 2 \log(|\mathcal{X}|\pi_t/\delta_f)$ with probability $1-\delta_f$,
        \[
        \text{abs}[f(\x) - \mu_{f,t}(\x)] \leq \beta^{1/2}_{f,t}\sigma_{f,t}(\x),  \hspace{5pt}\forall \x \in \mathcal{X}. \tag*{\qedsymbol}
        \]
    \end{lemma}
    \smallskip
    Next, we show that even though multiple GPs are used to model the slopes, the UCB bound still applies.

    \begin{lemma} \label{lemma:ucb_m}
        Suppose we have $n$ GPs $m_i$ over $\mathcal{X}$.
        For all $t \geq 1$ and $\beta_{m,t} = 2 \log(|\mathcal{X}|n\pi_t/\delta_m)$ with probability at least $1-\delta_m$, the following holds for all $i=1,\dots,n$:
        \[
        \text{abs}[f_i(\x) - \mu_{m_i,t}(\x)] \leq \beta^{1/2}_{m,t}\sigma_{m_i,t}(\x),  \hspace{5pt}\forall \x \in \mathcal{X}.
        \]
    \end{lemma}
    \medskip
    \begin{proof}
        Let $A_i$ be the event
        \[
        A_i = \{\text{abs}[m_i(\x) - \mu_{m_i,t}(\x)] \leq \beta^{1/2}_{m,t}\sigma_{m_i,t}(\x) \forall \x \in \mathcal{X}_i  \}.
        \]
        Then, $P[A_i^c] \leq |\mathcal{X}|\cdot e^{-\beta_{m,t}/2}$.
        Further,
        \begin{align*}
            P[(\cap_i A_i)^c] = P[\cup_i A_i^c] &\leq \sum_i P[A_i^c]\\
            &\leq \sum_i |\mathcal{X}|\cdot e^{-\beta_{m,t}/2}\\
            % &\leq e^{-\beta_{m,t}/2} \sum_i |\mathcal{X}|\\
            &\leq n|\mathcal{X}| e^{-\beta_{m,t}/2}.
        \end{align*}
        By applying DeMorgan's laws and the union bound, we obtain $P[\cap_i A_i] \geq 1 - |\mathcal{X}|n e^{-\beta_{m,t}/2}.$
        The remainder of the proof is identical to the proof of Lemma 5.1  in \cite{srinivas2012information}.
    \end{proof}

    \medskip
    We now establish properties of sets used in \safeslope.
    \begin{lemma} \label{lemma:sui_lem2} The following properties hold for all $t \geq 1$.
        \begin{enumerate}[(i)]
            % \item $\forall \x \in \mathcal{X}, u_{f,t+1}(\x) \leq u_{f,t}(\x)$. \label{lemma_enum:u}
            % \item $\forall \x \in \mathcal{X}, \ell_{f,t+1}(\x) \geq \ell_{f,t}(\x)$.
            % \item $\forall \x \in \mathcal{X}, w_{f,t+1}(\x) \leq w_{f,t}(\x)$. \label{lemma_enum:w}
            \item $S_{t+1} \supseteq S_t \supseteq S_0$. \label{lemma_enum:S}
            \item $S \subseteq D \implies R_{\epsilon, \hat{u}_t}(S) \subseteq R_{\epsilon, \hat{u}_t}(D)$.
            \item $S \subseteq D \implies \bar{R}_{\epsilon}(S) \subseteq \bar{R}_{\epsilon}(D)$.
            % \item $\forall \x \in \mathcal{X}, \hat{u}_{m_i,t+1}(\x) \leq \hat{u}_{m_i,t}(\x)$
        \end{enumerate}
    \end{lemma}
    \begin{proof}
        (i) From Lemma 2 of \cite{sui2015safe}, we know that (\ref{lemma_enum:S}) holds when the Lipschitz constant $L$ of $f(\cdot)$ is known. By replacing $L$ with $\hat{u}_{m_i, t}(x,x')$, it follows that for every $t \geq 1$ and given any $\x$, $\x'$,
        \begin{align*}
            &u_{f,t+1}(\x) + \hat{u}_{m_i,t+1}(\x,\x')\cdot d(\x,\x') \\
            &\leq u_{f,t}(\x) + \hat{u}_{m_i,t}(\x,\x')\cdot d(\x,\x') \leq h.
        \end{align*}
        From the definition of $u_{f,t}$ and $\hat{u}_{m_i,t}$, it follows that these bounds are non-increasing over time, for all $\x$. The second inequality follows from \eqref{eq:S_t}.
        Therefore, $S_{t+1} \supseteq S_t \supseteq S_0$.

        (ii) Let $\x \in R_{\epsilon, \hat{u}_t}(S).$ By definition of the reachability set, $\exists \x' \in S$ such that $f(\x')+\hat{u}_{m_i, t}\cdot d(\x,\x') + \epsilon \leq h$. As $S \subseteq D$, this implies $ \x' \in D$, which implies $\x \in R_{\epsilon, \hat{u}_t}$.

        (iii) This directly follows from repeatedly applying part (ii). Each reachability step is a union of two subsets of $\mathcal{X}$, so the union is bounded by $\mathcal{X}$ and the limit exists.
    \end{proof}

    \begin{remark}
        Note, by definition, the confidence bounds are always nested. Because the safe-set relies on these confidence bounds and not the true values, by definition, safe sets are also always nested. However, with probability less than $\delta_f$, $C_{f,t-1}(\x) \cap Q_{f,t}(\x) = \varnothing$, causing an poorly defined problem. While the following results theoretically hold with probability 1, the problem only remains defined with probability $1-\delta_f$.
    \end{remark}

    Next, we show that the width $w(x)$ is bounded by some $\epsilon >0$ using upper confidence bounds. Unlike \cite{sui2015safe, srinivas2012information}, we consider a non-unit variance for the kernel function $k$.
    \begin{lemma} \label{lemma:sui_cor2}
        Given a kernel with variance $v^2$ and measurement noise $\xi^2$, for each $t\geq 1$, define $T_t$ as the smallest positive integer satisfying $\frac{T_t}{\beta_{f,t+T_t}\gamma_{t+T_t}}\geq \frac{C_1}{\epsilon^2}$, where $C_1 = 8 v^2 /\log(1+v^2\xi^{-2})$. If $S_{t+T_t} = S_t$, then for any $\x \in G_{t+T_t} \cup M_{t+T_t}$, it holds that $w_{t+T_t}(\x) \leq \epsilon$. \hfill \qedsymbol
        % \end{align*}
\end{lemma}
\medskip

The proof follows the same steps as Lemma 5 in \cite{sui2015safe} and Lemma 5.4 in \cite{srinivas2012information} with the difference of a non-unit kernel variance. %Complete proof provided in \cite{lau2023multi}.
\medskip
\begin{proof}
    By definition, $w_t(\x_t) \leq 2\beta_t\sigma_{t-1}(\x_t)$. Similar to the steps in the proof of Lemma 5.4 in \cite{srinivas2012information},
    \begin{align*}
        w_t^2(\x_t)
        &\leq 4\beta_{f,t}^2\sigma_{t-1}^2(\x_t)\\
        &= 4\beta_{f,t}^2 \xi^{2}(\xi^{-2}\sigma_{t-1}^2(\x_t))\\
        &\leq 4\beta_{f,t}^2\xi^{2}C_2\log(1+\xi^{-2}\sigma_{t-1}^2(\x_t)),
    \end{align*}
    where $C_2 = (v^2\xi^{-2})/\log(1+v^2\xi^{-2}) \geq 1$.
    We leverage the inequalities $\xi^{-2}\sigma_{t-1}^2(\x_t) \leq\xi^{-2}k(\x_t,\x_t) \leq \xi^{-2}v^2$ and $s^2\leq C_2 \log(1+s^2)$ for $s^2 \in [0,\xi^{-2}]$ to obtain $C_2$. It is seen that $C_1 = 8 \xi^2 C_2$.

    Then
    \begin{align*}
        (t-t_0)w_t^2(\x_t) &\leq \sum_{\tau=t_0}^t w_\tau^2(\x_\tau) \tag*{by  nestedness of $w_t$}\\
        &\leq \frac{1}{2}C_1 \sum_{\tau=t_0}^t \beta_{f,\tau}^2\log(1+\xi^{-2}\sigma_{\tau-1}^2(\x_\tau))\\
        &\leq C_1\beta_{f,t}^2\gamma_t.
    \end{align*}
    The second inequality follows from Lemma \ref{lemma:sui_lem2}. The third inequality follows from $\beta_{f,t} > \beta_{f,\tau}$ for any $t > \tau$ and Lemma 5.3 of \cite{srinivas2012information}.
    As a result,
    \begin{align*}
        w_{t}(\x_t) \leq \sqrt{\frac{C_1 \beta_{f,t} \gamma_t}{t-t_0}}.
    \end{align*}
    Additionally, using the proposed condition on $T_t$, for any time $t+T_t$, $w_{t+T_t}(\x_{t+T_t}) \leq \epsilon$.
\end{proof}

In the following lemmas, we assume $C_1$ and $T_t$ are defined as in Lemma \ref{lemma:sui_cor2}.
We next establish guarantees on how $S_t$ evolves with time using the reachability operator.
\begin{lemma} \label{lemma:sui_lem7}
    For any $t \geq 1$, if $\bar{R}_{\epsilon}(S_0) \backslash S_t \neq \varnothing$, then with probability at least $1-\delta_f$,
    \begin{align} \label{eq:lem_sui7}
        S_{t+T_t} \supsetneq S_t.
    \end{align}
\end{lemma}

\begin{proof}
    We prove this by contradiction. First, for any $t \geq 1$, if $\bar{R}_{\epsilon}(S_t) \backslash S_t \neq \varnothing$, then $R_{\epsilon, \hat{u}_t}(S_t) \backslash S_t \neq \varnothing$ (by following steps identical to those in the proof of Lemma 6 in \cite{sui2015safe}). By the definition of $R_{\epsilon, \hat{u}_t}(S_t)$, we know that (a) $\exists \x' \in R_{\epsilon, \hat{u}_t}(S_t) \backslash S_t$ and (b) $\exists \x \in S_t$ so that
    \begin{align} \label{eq:reach}
        f(\x) + \epsilon + \hat{u}_{m_i, t}(\x,\x')\cdot d(\x,\x) \leq h.
    \end{align}

    Now, assume that contrary to \eqref{eq:lem_sui7}, $S_{t+T_t} = S_t$. This implies that $\x' \in V(S_{t+T_t}) \backslash S_{t+T_t}$ and $\x \in S_{t+T_t}$. As a result, with probability at least $1-\delta_f$,
    \begin{align*}
        &\ell_{f, t+T_t}(\x)+ \hat{u}_{m_i,t+T_t}(\x,\x')\cdot d(\x,\x') \\
        &\leq f(\x)+ \hat{u}_{m_i,t+T_t}(\x,\x')\cdot d(\x,\x')  \tag*{by Lemma \ref{lemma:ucb}} \\
        &\leq f(\x)+ \hat{u}_{m_i,t}(\x,\x')\cdot d(\x,\x')  \tag*{by \eqref{eq:u_hat}}\\
        &\leq  f(\x) + \epsilon + \hat{u}_{m_i, t}(\x,\x')\cdot d(\x,\x') \leq h \tag*{by \eqref{eq:reach}.}
    \end{align*}

    Therefore, $g_{t+T_t}(\x) > 0$ and $\x \in G_{t+T_t}$. Since we assumed that $S_{t+T_t} = S_t$ with $\x \in G_{t+T_t}$, we have
    \begin{align*}
        &u_{f,t+T_t}(\x) + \hat{u}_{m_i,t+T_t}(\x,\x') \cdot d(\x,\x') \\
        &\leq u_{f,t+T_t}(\x) + \hat{u}_{m_i,t}(\x,\x') \cdot d(\x,\x') \tag*{by \eqref{eq:u_hat}}\\
        &\leq u_{f,t+T_t}(\x) - f(\x) -\epsilon + h \tag*{by \eqref{eq:reach}}\\
        &\leq w_{t+T_t}(\x) - \epsilon + h \tag*{by Lemma \ref{lemma:ucb}}\\
        &\leq h \tag*{by Lemma \ref{lemma:sui_cor2}.}
    \end{align*}
    Eq. \eqref{eq:S_t} implies $\x' \in S_{t+T_t}$. This contradicts our assumption that $\x' \in V(S) \backslash S_{t+T_t}$. Therefore, $S_{t+T_t} \supsetneq S_t$.
\end{proof}

\begin{lemma} \label{lemma:sui_lem8}
    For any $t\geq 1$, if $S_{t+T_t} =S_t$, then with probability at least $1-\delta_f$,
    \begin{align*}
        f(\hat{\x}_{t+T_t}) \leq \min_{\x \in \bar{R}_\epsilon (S_0)}f(\x)+\epsilon.
    \end{align*}
\end{lemma}

\begin{proof}
    By solving a minimization rather than a maximization, the first part of the proof of Lemma 8 in \cite{sui2015safe} shows that $f(\hat{\x}_{t+T_t}) \leq f(\x^*) + \epsilon$, where $\x^*:=\arg\max_{\x\in S_{t+T_t}}f(\x)$.
    Then, since $S_{t+T_t} = S_t$, Lemma \ref{lemma:sui_lem7} implies that $\bar{R}_{\epsilon}(S_0) \subseteq S_t = S_{t+T_t}$. Therefore,
    \begin{align*}
        \min_{\x\in\bar{R}_\epsilon(S_0)} f(\x) +\epsilon &\geq \min_{\x\in S_{t+T_t}}f(\x) + \epsilon\\
        &= f(\x^*) + \epsilon \geq f(\hat{\x}_{t+T_t}).
    \end{align*}
\end{proof}

\begin{corollary} \label{cor:sui_cor3}
    For any $t\geq 1$, if $S_{t+T_t}=S_t$, then with probability at least $1-\delta_f$,
    \begin{align*}
        \forall t' \geq 0, f(\hat{\x}_{t+T_t+t'})\leq \min_{\x \in \bar{R}_{\epsilon} (S_0)}f(\x)+\epsilon. \tag*{\qedsymbol}
    \end{align*}

    Similar to the proof of Corollary 3 in \cite{sui2015safe}, this directly follows from Lemma \ref{lemma:sui_lem8}.
\end{corollary}

Having analyzed the evolution of the $S_t$, we now bound the time it takes to achieve the optimization goal.
\begin{lemma} \label{lemma:sui_lem10}
    Let $t^*$ be the smallest integer resulting in $t^* \geq |\bar{R}_{0}(S_0)|T_{t^*}$. Then, there exists a $t_0 \leq t^*$ such that $S_{t_0+T_{t_0}} = S_{t_0}$. \hfill \qedsymbol
\end{lemma}
The proof of this lemma is similar to the proofs of Lemma 9 and 10 in \cite{sui2015safe}, with the key difference of $R$ depending on the upper bound of $\hat{u}_t$ instead of a global constant $L$. %Complete proof provided in \cite{lau2023multi}.

\begin{proof}
    The first part of this proof is similar to the proof of Lemma 9 in \cite{sui2015safe}. We first show that for any $t\geq 0$, with probability at least $1-\delta_f$,
    \begin{align} \label{eq:sui_lem9}
        S_t \subseteq \bar{R}_{0}(S_0).
    \end{align}
    To show this, we use a proof by induction. For the base case, at $t=0$, by definition, $S_0 \subseteq \bar{R}_{0}(S_0)$.

    Now, we assume that for some $t\geq 1$, $S_{t-1} \subseteq \bar{R}_{0}(S_0)$. Let $\x \in S_t$. It suffices to show $\x \in \bar{R}_{0}(S_0)$. By the definition of $S_t$, $\exists \x' \in S_{t-1}$ so that with probability at least $1-e^{-\frac{1}{2}\beta_{f,t}}$,
    \begin{align*}
        u_{f,t}(\x')+\hat{u}_{m_i,t}(\x,\x') \cdot d(\x,\x') &\leq h\\
        \implies f(\x')+\hat{u}_{m_i,t}(\x,\x') \cdot d(\x,\x') &\leq h, \tag*{by Lemma \ref{lemma:ucb}.}
    \end{align*}

    Then $\x' \in \bar{R}_{0}(S_0)$ from our assumption. By the definition of $\bar{R}_{0}$, since $\x' \in \bar{R}_{0}(S_0)$, then $\x \in \bar{R}_{0}(S_0)$ and $S_t \subseteq \bar{R}_{0}(S_0)$.

    The remainder of the proof is similar to the steps of the proof of Lemma 10 in \cite{sui2015safe}.

    Contrary to our assertion, assume that for any $t \leq t^*$, $S_t \subsetneq S_{t+T_t}$. By Lemma \ref{lemma:sui_lem2} (\ref{lemma_enum:S}), we know that $S_t \subseteq S_{t+T_t}$ Then, since $T_t$ increases with t,
    \begin{align*}
        S_0 \subsetneq S_{T_0} \subseteq S_{T_t^*} \subsetneq S_{T_t^*+T_{T_t^*}} \subseteq S_{2T_t^*} \subsetneq \dots.
    \end{align*}
    This implies that for any $0 \leq p \leq |\bar{R}_{0}(S_0)|$, $|S_{kT_{t^*}}| > p$ and in particular, for $p^* := |\bar{R}_{0}(S_0)|$, we have
    \begin{align*}
        |S_{k^*T}| > |\bar{R}_{0}(S_0)|.
    \end{align*}
    This contradicts the \eqref{eq:sui_lem9}. Therefore, for the given $t^*$, $S_{t_0+T_{t_0}} = S_{t_0}$.
\end{proof}

\begin{corollary} \label{cor:sui_cor4}
    Let $t^*$ be the smallest integer resulting in $\frac{t^*}{\beta_{f,t^*} \gamma_t^*}\geq \frac{C_1\left(|\Bar{R}_{0}(S_0)|+1 \right)}{\epsilon^2}$. Then, there exists a $t_0 \leq t^*$ so that $S_{t_0+T_{t_0}} = S_{t_0}$.\hfill \qedsymbol

    The proof results directly from Lemmas \ref{lemma:sui_cor2} and \ref{lemma:sui_lem10}.
\end{corollary}

\textit{Proof of Theorem \ref{thm:ss}}:
For the first point of Theorem \ref{thm:ss}, the steps are similar to the proof of Lemma 11 in \cite{sui2015safe}. For the induction step, assume $f(\x)\leq h$ for some $t\geq 1$ and any $\x\in S_{t-1}$. Then, for any $\x \in S_t$, $\exists \x' \in S_{t-1}$ along some axis $i$ so that $h \geq u_{f,t}(\x') + \hat{u}_{m_i,t}(\x',\x) \cdot d(\x',\x)$.

With probability at least $1-e^{-\frac{1}{2}\beta_{f,t}}$,
\begin{align*}
    h&\geq f(\x') +\hat{u}_{m_i,t}(\x',\x) \cdot d(\x',\x) \tag*{by Lemma \ref{lemma:ucb}.}
\end{align*}

With probability at least $1-e^{-\frac{1}{2}\beta_{m,t}}$,
\begin{align*}
    &\geq f(\x') + m(\x',\x) \cdot d(\x',\x), \tag*{by Lemma \ref{lemma:ucb_m}} \\
    &\geq f(\x), \tag*{by the definition of $m$.}
\end{align*}
By applying the union bound across $|\mathcal{X}|$ realizations of $\x$, the resulting inequality holds with probability $1-\delta_f-\delta_m$.

The second point results from Corollaries \ref{cor:sui_cor3} and \ref{cor:sui_cor4}. \hfill \qedsymbol
\end{document}